%% file: main.tex
\documentclass[envcountsame]{elsarticle}
\usepackage{amssymb}
\usepackage{amsmath}
\usepackage{mathtools}
\usepackage{relsize}

\usepackage{framed}

\usepackage{url}
\usepackage{algpseudocode}
\usepackage{algorithm}
\usepackage{xcolor}
\usepackage{comment}
\usepackage{tikz}
\usetikzlibrary{shapes}

\colorlet{sjcolor}{blue}
\colorlet{bscolor}{green}
\colorlet{skcolor}{orange}
\colorlet{akcolor}{brown}

\newcommand{\bscomment}[1]{\textcolor{bscolor}{BS:#1}}
\newcommand{\skcomment}[1]{\textcolor{skcolor}{SK:#1}}

\newcommand{\Omit}[1]{}
\newcommand{\EM}[1]{{\em#1}}
\newcommand{\qkicoloring}[3]{\ensuremath{(#1,#2,#3)}}

\newcommand{\kneser}[2]{\ensuremath{K(#1,#2)}}

\newtheorem{theorem}{Theorem}
\newtheorem{lemma}[theorem]{Lemma}
\newdefinition{definition}[theorem]{Definition}

\newtheorem{proposition}[theorem]{Proposition}

\newproof{proof}{Proof}

\begin{document}

\title{On the Tractability 
of $(k,i)$-Coloring\tnoteref{t1}}
\tnotetext[t1]{This paper is the full version of the article ``On the Tractability 
of $(k,i)$-Coloring''\cite{caldam2018}, published in the CALDAM 2018 conference,
with the same set of authors.}
\author{Sriram Bhyravarapu}
\ead{cs16resch11001@iith.ac.in}
\author{Saurabh Joshi}
\ead{sbjoshi@iith.ac.in}
\author{Subrahmanyam Kalyanasundaram}
\ead{subruk@iith.ac.in}
\author{Anjeneya Swami Kare\corref{}}
\ead{cs14resch01002@iith.ac.in}
\address{Department of Computer Science and Engineering, IIT Hyderabad,\\Telangana - 502 285, India.}

\begin{abstract}
In an undirected graph, a proper $(k,i)$-coloring is an assignment of a set of $k$ colors to each vertex such that any two adjacent vertices 
have at most $i$ common colors. The $(k,i)$-coloring problem is to compute the minimum number of colors required for a 
proper $(k,i)$-coloring.
This is a generalization of the classic graph coloring problem.

We show a parameterized algorithm for the $(k,i)$-coloring problem with the size 
of the feedback vertex set as a parameter. Our algorithm does not use tree-width machinery,
thus answering a question of Majumdar, Neogi, Raman and Tale [CALDAM 2017]. We also give 
a faster and simpler exact algorithm for $(k, k-1)$-coloring.
From the hardness perspective, we
show that the $(k,i)$-coloring problem is NP-complete for  any fixed values $i, k$, whenever $i<k$, thereby settling
a conjecture of M\'{e}ndez-D\'{i}az and Zabala [1999] and again asked by Majumdar, Neogi, Raman and Tale.
The  NP-completeness result improves the partial NP-completeness shown in the preliminary version of 
this paper published in CALDAM 2018.
%
\end{abstract}

\maketitle

\input{Introduction}

\input{Preliminary}

\input{FPT-FVS}

\input{k-k-1}

\input{NP-completeness}

\section*{References}

\bibliographystyle{elsarticle-num}
\biboptions{sort&compress}
\bibliography{BibFile}
\end{document}

%% file: Introduction.tex


\section{Introduction}
\Omit{
Consider an undirected graph $G = (V,E)$ with $|V| = n$. A $(k,i)$-coloring is defined as an assignment of $k$ colors to every vertex such that any two adjacent vertices share at most $i$ colors. $(k,i)-chromatic number$, ${\chi_k^i}$ is the minimum number of colors needed for a proper $(k,i)-coloring$. This $(k,i)$-coloring problem was first introduced by M\'{e}ndez-Diaz and Zabala~\cite{IMDiaz1999}. This problem has been studied on a variety of graph classes, like complete graphs, perfect graphs, etc ~\cite{Bonomo2014}. Recently, Exact algorithms to determine the chromatic number for $(k,0)$-coloring and $(k,k-1)$-coloring have been studied by Diptapriyo Majumdar, Rian Neogi, Venkatesh Raman and Prafullkumar Tale~\cite{Diptapriyo2017}. They have also shown that $(k,i)$-coloring is FPT for Vertex cover and Tree width. There is a polynomial time algorithm to compute the chromatic number of a perfect graph. 

Given a graph $G$ and size of the Feedback Vertex Set $|S|$, we proposed an algorithm to prove that $(k,i)$-coloring is Fixed Parameterized Tractable(FPT) with Feedback Vertex Set(FVS) as a parameter. We also gave a $O(2^n * n^{O(1)})$ algorithm to determine the $(k,k-1)$-chromatic number which is better than known $O(2^n * n^{O(1)})$ algorithm ~\cite{Diptapriyo2017}.
}
In an undirected graph $G = (V, E)$, $|V| = n$, a \emph{proper vertex coloring} is to color the vertices of the graph such that adjacent vertices get different colors. The classic graph coloring problem asks to compute the minimum number of colors required to  properly color the graph. The minimum number of colors required is 
called the \emph{chromatic number} of the graph, denoted by $\chi(G)$. This is a well known NP-hard problem and has been studied in multiple directions.

Many variants and generalizations of the graph coloring problem have been studied in the past. In this paper we address a generalization of the graph coloring problem called $(k,i)$-coloring problem. For a proper $(k,i)$-coloring, we need to assign a set of $k$ colors to each vertex such that the adjacent vertices share at most $i$ colors. The $(k,i)$-coloring problem asks to compute the minimum number of colors required to properly $(k,i)$-color the graph. The minimum number of colors required is  called the $(k,i)$-\emph{chromatic number}, denoted by $\chi_{k}^{i}(G)$. Note that $(1,0)$-coloring is the same as the classic graph coloring problem.

\begin{framed}
\vspace*{-0.10cm}
\noindent $(k,i)$-{\sc Coloring Problem}\\
\textbf{Instance:} An undirected graph $G = (V, E)$. \\
\textbf{Output:} The $(k,i)$-chromatic number of $G$, $\chi_{k}^{i}(G)$.
\vspace*{-0.10cm}
\end{framed}
We also define below the $\qkicoloring{q}{k}{i}$-coloring problem, the decision version of the $(k,i)$-coloring problem.
\begin{framed}
\vspace*{-0.1cm}
\noindent $\qkicoloring{q}{k}{i}$-{\sc Coloring Problem}\\
\textbf{Instance:} An undirected graph $G = (V, E)$. \\
\textbf{Question:} Does $G$ have a proper $(k,i)$-coloring using at most $q$ colors? 
\vspace*{-0.1cm}
\end{framed}
The $(k,i)$-coloring problem was first studied by M\'{e}ndez-D\'{i}az and Zabala in~\cite{IMDiaz1999}. For arbitrary $k$ and $i$, the $(k,i)$-coloring problem is NP-hard because $(1,0)$-coloring is NP-hard. Apart from studying the basic properties, they also gave an integer linear programming formulation of the problem. 
Stahl~\cite{stahl1976} and independently Bollob\'{a}s and 
Thomason~\cite{Bollobas1979} introduced the $(k,0)$-coloring problem under the names of $k$-tuple coloring and $k$-set coloring
respectively. The $k$-tuple coloring problem has been studied in detail~\cite{klostermeyer2002,sparl2007}, and 
Irving~\cite{Irving1983} showed that this problem is NP-hard as well. Some of the applications for the $(k,0)$-coloring problem 
include construction of pseudorandom number generators, randomness extractors, secure password management schemes, aircraft scheduling, biprocessor tasks and frequency assignment to radio stations~\cite{marx2004,calvin2014}.
Brigham and Dutton~\cite{brigham1982} studied another variant of the problem, where $k$ colors have to be assigned 
to each vertex such that the adjacent vertices share exactly $i$ colors. 

Bonomo, Dur\'an, Koch and Valencia-Pabon~\cite{Bonomo2014}
studied the connection between the $(k, i)$-coloring problem on cliques and the theory of error correcting codes. 
In coding theory, a $(j, d, k)$-constant weight code represents a set of codewords of length $j$ with exactly $k$ 
ones in each codeword, with Hamming distance at least $d$. They observed a direct connection between $A(j, d, k)$, 
the largest possible 
size of a $(j, d, k)$-constant weight code, and the $(k,i)$-colorability of cliques and used the existing results from coding theory (such 
as the Johnson bound~\cite{johnson1962})
to infer results on the $(k,i)$-colorability of cliques. Finding bounds on $A(j, d, k)$ is a well-studied problem in coding theory, 
and lots of questions on $A(j, d, k)$ are still open. This indicates the difficulty of the $(k, i)$-coloring problem even on graphs as simple as cliques.

Since the $(k, i)$-coloring problem is NP-hard in general, it is natural to study the tractability for 
special classes of graphs. Polynomial time algorithms are only known for a few of such classes
namely bipartite graphs, cycles, cacti and graphs with bounded vertex cover or tree-width 
\cite{Bonomo2014,Diptapriyo2017}. From the NP-hardness perspective, it is interesting to ask
if the $(k,i)$-coloring problem is NP-hard for specific values of $i$. Except for the cases $i=k$, where 
the problem is trivial, and $i = 0$, where the problem is NP-hard~\cite{Irving1983}, the NP-hardness remains open
for all other values of $i$.

Recently, Majumdar, Neogi, Raman and Tale~\cite{Diptapriyo2017} studied the $(k,i)$-coloring problem 
and gave exact and parameterized algorithms
for the problem. They showed that the problem is fixed parameter tractable (FPT) 
when parameterized by tree-width.
As the tree-width is at most $(|S| + 1)$, where $S$ is a feedback vertex set (FVS) of the graph, their algorithm 
also implies that $(k,i)$-coloring is FPT when parameterized by the size of FVS.
As an open question, they asked to devise an FPT algorithm parameterized by the size of FVS, 
without going through tree-width. 

Our results are:
\begin{itemize}
\item An $O((^q_k)^{|S|+2}n^{O(1)})$ time algorithm for the $\qkicoloring{q}{k}{i}$-coloring problem 
that does not use tree-width machinery, 
Here $S$ is an FVS of the graph.
This implies
an FPT algorithm for the $(k,i)$-coloring problem parameterized by the size of FVS,
and thus answers the question
posed in \cite{Diptapriyo2017}.  

\item We show that the $(k,i)$-coloring problem is NP-complete for any fixed values $k, i$ whenever $i < k$. 
This answers questions posed in~\cite{IMDiaz1999} and~\cite{Diptapriyo2017} and settles the 
complexity of the $(k,i)$-coloring problem for all values of $k$ and $i$.
This result also improves the partial NP-completeness shown in the preliminary version~\cite{caldam2018} of 
this paper.

\item We give a $2^{n}n^{O(1)}$ time exact algorithm for the $(k,k-1)$-coloring problem. This is a direct improvement to the algorithm given in~\cite{Diptapriyo2017} for the same problem.
\end{itemize}

%% file: Preliminary.tex
\section{Preliminaries}
\label{prelims}
A \emph{parameterized problem} is a language $B \subseteq \Sigma^* \times \mathbb{N}$ where $\Sigma$ is a fixed, finite alphabet. For example $(x,\ell) \in \Sigma^* \times \mathbb{N}$, here $\ell$ is called the parameter. A parameterized problem $B \subseteq \Sigma^* \times \mathbb{N}$ is called \emph{fixed-parameter tractable} (FPT) if there is an algorithm $\mathcal{A}$, a computable function $f:\mathbb{N} \to \mathbb{N}$, and a constant $c$ such that, given $(x,\ell) \in \Sigma^* \times \mathbb{N}$, the algorithm $\mathcal{A}$ correctly decides whether $(x,\ell)\in B$ in time bounded by $f(\ell)|x|^c$.

We assume that the graph is simple and undirected. We use $n$ to denote $|V|$, the number of 
vertices of the graph.
We say that the vertices $u$ and $v$ are \emph{adjacent} (\emph{neighbors}) 
if $\{u,v\}\in E$. 
For $v \in V$, we let $N(v)$ denote the set of neighbors of $v$.
For $S \subseteq V$, the sub graph induced by $S$ is denoted by $G[S]$. We use $O^{*}(f(n))$ to denote $O(f(n) n^{O(1)})$.
We use the set of natural numbers for coloring the graph. We use the standard notations $[q] = \{1, 2, \ldots, q\}$ and $\binom{[q]}{k}$ to 
denote the set of all $k$-sized subsets of $[q]$.
In the rest of the paper, we use  the term  \emph{coloring} of a set $X \subseteq V$ to denote a mapping $h: X \rightarrow \binom{[q]}{k}$.
We say that $h$ is a \emph{proper} $\qkicoloring{q}{k}{i}$-coloring (or proper $(k,i)$-coloring) of $X$ if any pair of adjacent vertices in $X$ have no more than $i$ colors in common.

%% file: FPT-FVS.tex
\section{$\qkicoloring{q}{k}{i}$-Coloring Parameterized by Size of FVS}
\label{tw}

In this section, we assume that $q, k, i$ are fixed values and focus on the decision problem of
$\qkicoloring{q}{k}{i}$-coloring.
A \EM{Feedback Vertex Set (FVS)} is a set of vertices $S \subseteq V$, removal of which from the graph $G$ makes the
remaining graph ($G[V\setminus S]$) acyclic. 
Many NP-hard problems have been shown to be tractable for graphs with bounded FVS~\cite{kratsch2010}.

In~\cite{Diptapriyo2017}, Majumdar, Neogi, Raman and Tale gave an 
$O((^q_k)^{tw+1} n^{O(1)})$  time\footnote{Even though
\cite{Diptapriyo2017} claims a running time of $O((^q_k)^{tw} n^{O(1)})$ for their algorithm, 
there is an additional factor of $\binom{q}{k}$ that is omitted, presumably because $\binom{q}{k}$ is treated as a constant.}
algorithm for the
$\qkicoloring{q}{k}{i}$-coloring problem,
where $tw$
denotes the tree-width of the graph. 
Let $S$ be a smallest FVS of $G$.
It is known that $tw \leq |S| +1$, see for instance~\cite{jansen2014}. In
this section, we present an algorithm for $\qkicoloring{q}{k}{i}$-coloring that runs in
$O((^q_k)^{|S|+2} n^{O(1)})$ time, where $|S|$ is the size of the FVS of the graph. Our
algorithm does not use the tree-width machinery. Note that, FVS has a $2$-approximation
algorithm~\cite{fvsapprox}, but 
there is no known polynomial time algorithm that approximates tree-width within a constant factor~\cite{inapproxtw}.
Computing the size of the smallest FVS is also known to be FPT parameterized by $|S|$, the size
of the smallest FVS. There has been a series of results improving the running time, the fastest known 
algorithm \cite{fastfvs}
runs in $O(3.619^{|S|} n^{O(1)})$ time.

A brief description of our algorithm follows. Let $S$ be an FVS of $G$. 
We start with a coloring 
of the vertices of $S$. Recall that $G[V\backslash S]$ is a forest. Each of the connected components of 
$G[V\backslash S]$ is a tree. For each of these components, we traverse the tree bottom-up and use a dynamic
programming technique to compute the list of $k$-colorings that each vertex $w \in V\backslash S$ can take. 
For each $C \in \binom{[q]}{k}$,
we include $C$ in $w$'s list if there is a coloring for the 
subtree rooted at $w$, consistent with the coloring of $S$, such that $w$ receives color set $C$.
We repeat this for all proper colorings of $S$.

Let $\Psi = \binom{[q]}{k}$ denote the family of all $k$-sized 
subsets of $[q]$.
For any pair of sets $C, C' \in \Psi$, we say that $(C, C')$ is \emph{legal} if $|C \cap C'| \leq i$, and \emph{illegal} 
if $|C \cap C'| > i$.
Given two sets $C, C' \in \Psi$, it is easy to check if $(C, C')$ is a legal pair. Formally, we have:
\begin{proposition}
\label{prop:legalcheck}
Given $C,C' \in \Psi$, 
it takes $O(k \log k)$ time to check if $(C,C')$ is a legal pair. 
\end{proposition}

\begin{definition}\label{def:hcompat}
Consider a partial coloring $h: S \rightarrow \Psi$ where only the vertices of the FVS $S$ are colored. 
For a vertex $w \in V \backslash S$ and a set $C \in \Psi$, we say that $(w, C)$ is \emph{$h$-compatible} if
for all $x \in S \cap N(w)$, the pair $(C, h(x))$ is legal.

The set $\{C \in \Psi\;|\; (w, C) \mbox{ is } h\mbox{-compatible}\}$ is defined to be the set of \emph{$h$-compatible colorings of $w$}.
\end{definition}

\begin{proposition}
\label{prop:compat}
Let $h: S \rightarrow \Psi$ be a coloring of the vertices in $S$.
Let $w \in V \backslash S$ and $d_S(w) = |N(w) \cap S|$. 
Then the set of $h$-compatible colorings of $w$ can be computed in time $O\left({q \choose k}d_S(w)k \log k\right)$.
\end{proposition}
\begin{proof}
For each $C \in \Psi$, we check if $(w, C)$ is $h$-compatible. For this, we need to check for all neighbors $x$ of $w$ in $S$, 
whether $(C, h(x))$ is legal. The total running time is $\binom{q}{k} \cdot d_S(w) \cdot O(k \log k)$. 
\qed
\end{proof}

\begin{definition}\label{def:extn}
Given a graph $G = (V, E)$ and a coloring $h: X \rightarrow \Psi$ for some $X \subseteq V$, 
we say that the coloring $h': V \rightarrow \Psi$ 
is an \emph{extension of $h$}, or \emph{extends $h$} if for all $v \in X$, we have $h(v) = h'(v)$.
\end{definition}

\begin{lemma}
Given a proper
$\qkicoloring{q}{k}{i}$-coloring $h$ of the vertices in a feedback vertex set $S$
of the graph $G=(V,E)$, we can
determine if $h$ can be extended to a 
proper
$\qkicoloring{q}{k}{i}$-coloring of $V$ in $O({q \choose k}^2 n^{O(1)})$ time.
\label{lemfvs}
\end{lemma}
\begin{proof} 
The graph $G[V\setminus S]$ is a forest because $S$ is a feedback vertex set.
Therefore each connected component of $G[V\setminus S]$ is a tree.
Below, we describe an algorithm that we can apply to each of these trees to yield a 
proper $\qkicoloring{q}{k}{i}$-coloring extending $h$ for the trees.
Combining the colorings, we get a proper $\qkicoloring{q}{k}{i}$-coloring of $V$, that is an extension of $h$.

Let $T$ denote one of the trees in the forest.
We will designate any one of the vertices (say $r$) of $T$ as root. Let $T_w$ denote
the subtree rooted at a node $w \in T$. 

Our plan is to maintain a table at each vertex $w$, indexed with the elements of $\Psi$.
The entry at each color set $C$ is denoted by $M_w(C)$. The entry $M_w(C)$ indicates 
whether there is a proper $\qkicoloring{q}{k}{i}$-coloring of $T_w$, with $w$ assigned the set $C$,
consistent with the coloring $h$ of  $S$.

We will process $T$ in a post order fashion as follows:
\begin{enumerate}
\item \textbf{When $w$ is a leaf in $T$:} In this case, we set $M_w(C) = 1$ if $(w, C)$ is $h$-compatible. Otherwise, 
we set  $M_w(C) = 0$.

For any leaf $w$, the values $M_w(C)$ corresponding to all $C \in \Psi$ 
can be computed in time $O({q \choose k}d_S(w)k \log k)$
by Proposition \ref{prop:compat}. Here $d_S(w)$ denotes the number of neighbors of $w$ in $S$.
\item \textbf{When $w$ is an internal node in $T$:} 
Let $u_1, u_2, \ldots$ be the children of $w$ in $T$. 
Recall that we process $T$ in post order fashion. 
Before we process $w$,  the $M_{u_j}$ values for all the children of $w$ would already have been computed.
The value $M_w(C)$ is computed as follows:
\begin{itemize}
	\item If $(w, C)$ is not $h$-compatible, we set $M_w(C) = 0$. 
	\item If $(w, C)$ is $h$-compatible, we do the following: 
	\begin{itemize}
		\item[--] If for each child $u_j$ of $w$, there exists at least one coloring $C' \in \Psi$ such that $M_{u_j}(C') = 1$
		 and $(C, C')$ is a legal pair, then set $M_w(C) = 1$.
		\item[--] Otherwise set $M_w(C) = 0$.
	\end{itemize}
\end{itemize}
For each $w$ and $C$, the $h$-compatibility check takes $O(d_S(w) k \log k)$
time. If $(w, C)$ is $h$-compatible, 
we need to check all the children $u_j$, and the table entries $M_{u_j}(C')$ for all $C' \in \Psi$. Together with the
check for $(C, C')$ being a legal pair, the computation takes $d_T(w) \cdot \binom{q}{k} \cdot 
O(k \log k)$
time, where
$d_T(w)$ is the number of children of $w$ in the tree $T$.

Adding all up, the computation of the table entries for $w$ takes time
\begin{equation}\label{mwtime}
O\left(\binom{q}{k} \cdot k \log k \cdot
\left[d_S(w) + d_T(w) \binom{q}{k}\right]\right).
\end{equation}
\end{enumerate}
If for some $C\in \Psi$, $M_r(C)=1$, then we know that there exists a proper $\qkicoloring{q}{k}{i}$-coloring of $T$ 
that is consistent with the coloring $h$ of $S$.

The time complexity is obtained by adding the expression in (\ref{mwtime}) over all the vertices $w \in V\backslash S$. 
By using the bounds $d_S(w) \leq n$ and $\sum_{\text{Trees }T} \sum_{w \in V(T)} d_T(w) \leq \sum_{\text{Trees }T} |V(T)| \leq n$, we get that the time complexity 
is upper bounded by
$$
O\left(\binom{q}{k} \cdot k \log k
\cdot \left[n^2 + n \binom{q}{k}\right]\right),
$$
which is at most $O\left(\binom{q}{k}^2 \cdot n^2\right)$, by noting that $k$ is a constant. 
\qed
\end{proof}
The correctness of the procedure explained in the above lemma can be proved using 
an induction on the vertices of $T$ according to its post order traversal. The inductive claim
says that $M_w(C) = 1$ if and only if there is a proper $\qkicoloring{q}{k}{i}$-coloring of $T_w$, 
with $w$ assigned the set $C$,
consistent with the given coloring of $S$.

\begin{lemma}\label{lem:qkispace}
Given a proper
$\qkicoloring{q}{k}{i}$-coloring $h$ of the vertices in a feedback vertex set $S$
of the graph $G=(V,E)$, we can
determine if $h$ can be extended to a proper $\qkicoloring{q}{k}{i}$-coloring of $V$ with space complexity 
$O({q \choose k} n)$.
\end{lemma}
\begin{proof} 
Recall the algorithm explained in Lemma \ref{lemfvs}.
At each vertex $w$ in $ G[V \setminus S]$, we need $O({q \choose k})$ space to store values $M_w(C)$ for all $C \in \Psi$.
\qed
\end{proof}



\begin{theorem}\label{thm:qki}
The $\qkicoloring{q}{k}{i}$-coloring problem can be solved in time 
$O((^q_k)^{|S|+2} n^{O(1)})$ and $O({q \choose k} n)$ space, where $S$ is 
a feedback vertex set of $G$.
\end{theorem}

\begin{proof}
For each coloring assignment $h$ of $S$, 
we first determine if $h$ is a proper $\qkicoloring{q}{k}{i}$-coloring.
This can be done in $O(|S|^2 k \log k)$ time.
Then we determine whether there exists a proper $\qkicoloring{q}{k}{i}$-coloring that extends $h$
in $O((^q_k)^2 . n^{O(1)})$ time by Lemma~\ref{lemfvs}. 
Since there are at most $(^q_k)^{|S|}$ many colorings of $S$, we can determine 
whether there exists a proper $\qkicoloring{q}{k}{i}$-coloring of $G$ in $O((^q_k)^{|S|+2} n^{O(1)})$ time.

We need $O(|S| k \log q)$ space to store the coloring $h$ of $S$. And by Lemma \ref{lem:qkispace},
we need $O({q \choose k} n)$ space to determine if $h$ can be extended to a proper coloring of $G$.
The latter is the dominating term and determines the total space requirement of the algorithm.
\qed
\end{proof}

\noindent\textbf{On generating a proper $\qkicoloring{q}{k}{i}$-coloring.} We 
observe that we can modify Theorem \ref{thm:qki} to obtain
an algorithm that generates a proper $\qkicoloring{q}{k}{i}$-coloring of $G$, if one exists.  
After executing the steps of the algorithm corresponding to Theorem \ref{thm:qki}, 
we traverse the tree in top-down fashion from the root, and find
colorings for each vertex $w \in T$, consistent with its parent, subtree $T_w$ and coloring of $S$.
The latter two are already encoded in $M_w(C)$ value. The asymptotic time and space
complexity are the same as that in Theorem \ref{thm:qki}.

We would like to observe a difference in the space usage of our FPT algorithm 
to the FPT algorithm for $\qkicoloring{q}{k}{i}$-coloring parameterized 
by tree-width in \cite{Diptapriyo2017}.
We note that the algorithm in \cite{Diptapriyo2017} can also be modified similarly to obtain an algorithm that generates 
a proper coloring. However, such an algorithm would require to store all feasible colorings 
at each bag of the tree-decomposition, resulting in a $O((^q_k)^{tw+1})$ space usage
at each bag. Since there are $O(n)$ bags, total space required by the algorithm is $O((^q_k)^{tw+1} n)$,
which is significantly larger than the $O({q \choose k} n)$ space required by our algorithm.

\medskip
\noindent\textbf{Decision vs. search problem.} 
We now note that $\chi_k^i(G) \leq \chi_k^i(G[S]) + \chi_k^i(G[V\backslash S])$.
In the RHS, the first term $\chi_k^i(G[S]) \leq k |S|$ trivially, and the second term 
$\chi_k^i(G[V\backslash S]) \leq 2k -i$ since $G[V \backslash S]$ is a forest. Thus we
have $\chi_k^i(G) \leq k |S| + 2k - i \leq k (|S| + 2)$.

We note that we could run the algorithm for
$\qkicoloring{q}{k}{i}$-coloring 
and perform binary search between 1 and $k (|S| + 2)$ and determine 
$\chi_k^i(G)$, the smallest 
$q$ for which the graph has a proper $\qkicoloring{q}{k}{i}$-coloring. 
The running time of this procedure would be at most $\log (k (|S|+2))$
times the running time of the algorithm presented in Theorem \ref{thm:qki},
that is 
$\log (k(|S| +2)) \cdot O(\binom{k (|S| +2)}{k}^{|S|+2} n^{O(1)})$.

Thus the FPT algorithm parameterized 
by the size of the FVS for the $\qkicoloring{q}{k}{i}$-coloring problem implies an 
FPT algorithm parameterized by the size of the FVS for the $(k,i)$-coloring problem as well.

\subsection{Counting all proper $\qkicoloring{q}{k}{i}$-colorings}
Here we show that we can modify the algorithm described in Lemma \ref{lemfvs} to count 
the number of proper $\qkicoloring{q}{k}{i}$-colorings of $G$.
Let a proper $\qkicoloring{q}{k}{i}$-coloring $h$ of FVS $S$ be given. 
Instead of maintaining $M_w(C)$ for a vertex $w$ in a rooted tree $T$, we maintain another value $M^\#_w(C)$.


\begin{equation*}
 M^\#_w(C)= \left\{ \begin{aligned}[lr]
    &0 & \text{if $(w, C)$ is not $h$-compatible.} \\
    &1 & \left\{ \begin{aligned}[l] & \text{if $w$ is a leaf,} \\ & \text{and $(w, C)$ is $h$-compatible.} \end{aligned} \right.\\
    &\prod\limits_{\forall u_j \in \text{{\scriptsize child}}(w)}     {\sum\limits_{{\text{{\scriptsize legal}}(C,C')}} M^\#_{u_j}(C')} & 
    \left\{ \begin{aligned}[l] &  \text{if $w$ is a non-leaf vertex,} \\ & \text{and $(w, C)$ is $h$-compatible.} \end{aligned} \right.\\
  \end{aligned} \right.
\end{equation*}


At each vertex $w$, $M^\#_w(C)$ maintains a count of the proper $\qkicoloring{q}{k}{i}$-colorings 
of $T_w$, consistent with the coloring $h$ of $S$, where $w$ gets assigned the set $C$.
The correctness can be verified by a straightforward induction on the tree vertices in post order traversal.
If $r$ is the root of $T$, $M^\#_r(C)$ gives the count of proper $\qkicoloring{q}{k}{i}$-colorings of $T$, where $r$ 
is colored $C$, consistent
with the coloring $h$ of $S$.

The total number of proper $\qkicoloring{q}{k}{i}$-colorings of $G$ is therefore computed by taking into
account (i) all proper $\qkicoloring{q}{k}{i}$-colorings $h$ of $S$, (ii) all the trees $T_j$ in 
$G[V\backslash S]$, and (iii) all color sets $C \in \Psi$ at the root of $T_j$.
The full expression is as follows:
\begin{equation*}
\text{No. of proper } \qkicoloring{q}{k}{i}\text{-colorings }= 
\sum_{\substack {\text{{\scriptsize proper }}\qkicoloring{q}{k}{i}\text{{\scriptsize-}}\\ \text{{\scriptsize colorings of }} S}} 
\left(
\prod\limits_{\text{{\scriptsize  }} T_j \text{{\scriptsize  in }} G[V \backslash S]} 
\left( \sum\limits_{C\in \Psi}{M^{\#}_{\text{{\scriptsize root(}}T_j\text{{\scriptsize )}} }}(C) \right)
\right).\end{equation*}
The above expression implies the following theorem. The asymptotic time complexity remains the same as 
Theorem \ref{thm:qki}, whereas the space complexity incurs a blowup of $n k \log q$, because of the maximum
value $M^\#_w(C)$ can take.

\begin{theorem}\label{thm:qkicount}
There is an algorithm that computes the number of proper $\qkicoloring{q}{k}{i}$-colorings 
of $G$, in $O((^q_k)^{|S|+2} n^{O(1)})$ time and $O({q \choose k} n^2 \log q)$ space, where $S$ is 
a feedback vertex set of $G$.
\end{theorem}

%% file: k-k-1.tex
\section{Faster Exact Algorithm for $(k,k-1)$-coloring}\label{sec:kk}

The article~\cite{Diptapriyo2017} gave an $O^*(4^n)$ time exact algorithm for the $(k,k-1)$-coloring problem. 
Their algorithm was based on running an exact algorithm for a set cover instance
where the universe is the set of all the vertices $V$ and the family of sets $\mathcal F$ is
the set of all independent sets of vertices of $G$.
To show correctness and running time, they used a claim (unnumbered) that relates $\chi_k^{k-1}(G)$
to the size of solution of the set cover instance, 
an $O(2^n\cdot n\cdot  |\mathcal F|)$ time exact algorithm for the set cover problem~\cite{Fomin2010} and 
an upper bound of  $2^n$ on the size of the family of sets $\mathcal F$.
Hence, the time complexity of their algorithm is $O(2^n\cdot n\cdot 2^n) = O(4^n\cdot n)$.

We first note that their algorithm also works when $\mathcal F$ is replaced by
$\mathcal F'$, the set of all maximal independent sets of $G$.
This is because any independent set $A \in \mathcal F$ is contained in a maximal independent set
$A' \in \mathcal F'$. In any set covering of $V$ using elements of $\mathcal F$, each set $A$ can be replaced by 
an $A'\in \mathcal F'$, thus obtaining a set cover of $V$ using elements of only $\mathcal F'$.
By using the $3^{n/3}$ upper bound of Moon and Moser \cite{moonmoser} on the
number of maximal independent sets, the time complexity improves to $O(2^n\cdot n\cdot 3^{n/3}) 
= O(2.88^n \cdot n)$. 

We now present a 
simpler and faster
$O^*(2^n)$ algorithm to determine $\chi_k^{k-1}(G)$.

\begin{lemma}\label{lemk,k-1}
For any graph $G$, $\chi_k^{k-1}(G) = q$ where $q$ is the smallest integer such that $\binom{q}{k} \geq \chi_1^0(G)$. 
Thus there is a polynomial time reduction from the $(k, k-1)$-coloring problem to the $(1, 0)$-coloring problem.
\end{lemma}
\begin{proof}
The $(k, k-1)$-coloring problem asks to assign sets of $k$ colors to each vertex, with the requirement 
that neighboring vertices must have distinct sets assigned to them. We may view each of the $k$-sized 
subsets
as a color, and the $(1,0)$-chromatic number $\chi_1^0(G)$ is the number of distinct $k$-sized subsets required.

Thus $\chi_k^{k-1}(G)$ is the smallest $q$ that will provide $\chi_1^0(G)$ number of $k$-sized subsets.
The polynomial time reduction is immediate.
\qed
\end{proof}

Combining the above lemma with the $O^*(2^n)$ time algorithm of Koivisto \cite{Koivisto} to compute $\chi_1^0(G)$, we
get the following theorem.

\begin{theorem}
There is an algorithm with $O^*(2^n)$ time complexity that computes the $(k,k-1)$ chromatic number of a given graph.
\end{theorem}

Further, we can infer from Lemma \ref{lemk,k-1} that for those graphs $G$  where we can compute $\chi _{1}^{0}$(G) 
in polynomial time, $\chi _{k}^{k-1}$($G$) can also be found in polynomial time. For instance, $\chi _{k}^{k-1}(K_n)$  can be computed in polynomial time as $\chi _{1}^{0}(K_n)$ = $n$.

%% file: NP-completeness.tex
\section{NP-completeness results}
Since the $(k,i)$-coloring problem is a generalization of the $(1,0)$-coloring problem, 
it follows that $(k, i)$-coloring is NP-hard in general.
Notice that we have $\chi_k^k(G) = k$ for all graphs $G$. Thus the 
$(k,k)$-coloring problem is trivial. 
M\'{e}ndez-D\'{i}az and Zabala  \cite{IMDiaz1999} conjectured that
the $(k, i)$-coloring problem is NP-hard whenever $i<k$.
In this section, we prove their conjecture by
showing that the  $(k, i)$-coloring problem is NP-complete for all values
of $k$ and $i$, as long as $i<k$.
Given a coloring, we can easily verify
that it is a proper $(k,i)$-coloring in polynomial time.
Therefore, we will only be proving the NP-hardness aspect of NP-completeness.

\subsection{Simple proofs for $(k,1)$-coloring and $(k, k-1)$-coloring}
In this section, we provide simple proofs 
for the NP-completeness of $(k,1)$-coloring and $(k, k-1)$-coloring.
For the $(k,0)$-coloring problem,
we have the following result by Irving. 

\begin{theorem}[$(k, 0)$-coloring is NP-complete \cite{Irving1983}]\label{thm:kzero}
The $\qkicoloring{2k+1}{k}{0}$-coloring problem is NP-complete for all $k \geq 1$.
\end{theorem}

The NP-completeness of the $(k,k-1)$-coloring problem is claimed by \cite{IMDiaz1999}. However, 
we are unable to follow and verify the proof. We provide an alternate NP-hardness proof 
as a consequence of the correspondence in Lemma \ref{lemk,k-1}.

\begin{theorem}\label{thm:npkminusone}
The $(k, k-1)$-coloring problem is NP-complete for all $k \geq 1$.
\end{theorem}
\begin{proof}
We use reductions from the $(1, 0)$-coloring problem, for each value of $k \geq 2$.
We show that
the $\qkicoloring{q}{k}{k-1}$-coloring problem is NP-complete for all values of $q > k \geq 2$.
From the correspondence in Lemma \ref{lemk,k-1}, it follows that for any given $k \geq 1$, 
a graph $G$ is $\qkicoloring{q}{k}{k-1}$-colorable if and only if $G$ is
$\qkicoloring{\binom{q}{k}}{1}{0}$-colorable.
Since the $\qkicoloring{r}{1}{0}$-coloring problems are NP-complete for all $r \geq 3$, 
it follows that $\qkicoloring{\binom{q}{k}}{1}{0}$-coloring problems are NP-complete for all $q >k \geq 2$, and hence
we get that the $\qkicoloring{q}{k}{k-1}$-coloring problems are NP-complete for all $q > k \geq 2$.
\qed
\end{proof}


The following lemmas will help us in proving further NP-completeness results.

\begin{lemma}[Complement trick]\label{lem:comp}
For integers $k, i\geq 1$, any graph $G$ is $\qkicoloring{2k+i}{k+i}{i}$-colorable if and only if
it is $\qkicoloring{2k+i}{k}{0}$-colorable.
\end{lemma}
\begin{proof}
Let $f: V \rightarrow \binom{[2k + i]}{k}$ be a $\qkicoloring{2k+i}{k}{0}$-coloring of $G$. Consider the coloring $f'$ where each 
vertex $v$ is assigned the complement set $[2k + i] \backslash f(v)$. Notice that, every vertex is assigned $(k+i)$ colors, 
and any pair of adjacent vertices will share exactly $i$ colors in the coloring $f'$. Thus we have a $\qkicoloring{2k+i}{k+i}{i}$-coloring
of $G$.

Similarly, if we start from a $\qkicoloring{2k+i}{k+i}{i}$-coloring of $G$, we can get to a $\qkicoloring{2k+i}{k}{0}$-coloring 
by taking the complement coloring. 
\qed
\end{proof} 

Theorem \ref{thm:kzero} and the above lemma together imply the NP-completeness of $(k, 1)$-coloring
for all $k \geq 2$.

\begin{theorem}[$(k, 1)$-coloring is NP-complete]\label{thm:kone}
The $\qkicoloring{2k+1}{k+1}{1}$-coloring problem is NP-complete for all $k \geq 1$.
\end{theorem}

\subsection{NP completeness of $(k, i)$-coloring}
In this section, we prove the NP-completeness of $(k,i)$-coloring.  The result is stated below:
\begin{theorem}[$(k, i)$-coloring is NP-complete]\label{thm:gennpc}
The $(k,i)$-coloring problem is NP-complete when $i < k$. That is, for any fixed values $i, k$ such that $i < k$, there exists a $q$ such that 
the $\qkicoloring{q}{k}{i}$-coloring problem is NP-complete.
\end{theorem}
The main ingredient of the above NP-completeness result is the following theorem, which generalizes 
Theorem \ref{thm:kzero} (originally proved in \cite{Irving1983}).
\begin{theorem}\label{thm:npcing}
The $\qkicoloring{2k + i}{k}{0}$-coloring problem is NP-complete for all $k, i \geq 1$.
\end{theorem}
Most of the remaining part of the section will be used to prove Theorem \ref{thm:npcing}. We first show how to 
infer Theorem \ref{thm:gennpc} from Theorem \ref{thm:npcing}.

\begin{proof}[Proof of Theorem \ref{thm:gennpc}]
Let $k', i$ be such that $k' > i \geq 0$. We will show that $(k',i)$-coloring is NP-complete.
Theorem \ref{thm:npcing} shows the NP-completeness of $(k',0)$-coloring when $k' \geq 1$.

To show the NP-completeness of $(k', i)$-coloring when $k' > i \geq 1$, let $k = k' - i$. Notice
that $k, i \geq 1$.
By Theorem \ref{thm:npcing}, we have that $\qkicoloring{2k + i}{k}{0}$-coloring is NP-complete. By  
Lemma \ref{lem:comp} (complement trick), a graph is $\qkicoloring{2k + i}{k}{0}$-colorable if and only if
it is $\qkicoloring{2k+i}{k+i}{i}$-colorable. Hence the $\qkicoloring{2k+i}{k+i}{i}$-coloring problem is 
NP-complete as well. Substituting $k' = k + i$, we get that the $\qkicoloring{2k' - i}{k'}{i}$-coloring
problem is NP-complete.
\qed
\end{proof} 

Now we shall start working towards proving Theorem \ref{thm:npcing}. The \emph{Kneser graph} 
forms an important component of the proof.
\begin{definition}[Kneser Graph]
The Kneser graph $\kneser{r}{k}$ is the graph whose vertices are $\binom{[r]}{k}$, the $k$-sized subsets of
$[r]$, and vertices $x$ and $y$ are adjacent if and only if $x \cap y = \emptyset$ (when $x$ and $y$ are 
viewed as sets).
\end{definition}
Consider the Kneser graph $\kneser{r}{k}$. We may view the elements of the set $[r]$ as colors. 
In this case, the sets associated with the vertices themselves form
a proper $(k,0)$-coloring of $\kneser{r}{k}$.  We will call this $(k,0)$-coloring
as the \emph{natural coloring}, denoted by $N$.
This is because two vertices are adjacent if and only if they do not share any colors. 

The following theorem
states that when $r \geq 2k + 1$, the natural coloring is essentially the only proper $(k,0)$-coloring
of the Kneser graph $\kneser{r}{k}$ that uses $r$ colors. 

\begin{theorem}\label{thm:natural}
Let $r \geq 2k + 1$ and $k \geq 1$. Any $\qkicoloring{r}{k}{0}$-coloring $C$ of $\kneser{r}{k}$ can be obtained from a natural coloring $N$ by a permutation of colors.
\end{theorem}

Before proving the above theorem, we state two known results that would be necessary in the proof.

\begin{theorem}[Erd\H{o}s-Ko-Rado Theorem \cite{ekr}]\label{thm:ekr}
Let $r \geq 2k$. The largest independent set of the Kneser graph $\kneser{r}{k}$ is of size $\binom{r-1}{k-1}$.
\end{theorem}

\begin{theorem}[Hilton and Milner \cite{hm}]\label{thm:hilton}
Let $r > 2k$. Every independent set of the Kneser graph $\kneser{r}{k}$ of size $\binom{r-1}{k-1}$
are the set of vertices that contain some color $a$ in their natural coloring.
\end{theorem}

We first require the following lemma.

\begin{lemma}\label{lem:colorpat}
Let $r \geq 2k$ and $k \geq 1$. Any $(k,0)$-coloring of $\kneser{r}{k}$ requires at least $r$ colors. If $C$ is a  $\qkicoloring{r}{k}{0}$-coloring of 
$\kneser{r}{k}$, then each color in $C$ must occur exactly $\binom{r-1}{k-1}$ times.
\end{lemma}
\begin{proof}
Let $N$ be the natural coloring and $C$ be an arbitrary $\qkicoloring{r}{k}{0}$-coloring of 
$\kneser{r}{k}$.
Each color occurs exactly $\binom{r-1}{k-1}$ times in $N$. 
Below, we argue that this must happen in $C$ as well. 

All the vertices containing a specific color form an independent set.
If any color in $C$  occurs more than $\binom{r-1}{k-1}$ times, then we will have a contradiction 
to Theorem \ref{thm:ekr}.
Now notice that the total number of ``color slots'' is fixed at $k \cdot \binom{r}{k}$. 
Hence each color must occur $k\cdot \binom{r}{k}/r = \binom{r-1}{k-1}$ times on average.
 If any color in $C$ occurs strictly less than $\binom{r-1}{k-1}$ times, then some other color has to occur 
 more than $\binom{r-1}{k-1}$ times in order to compensate. Hence each color in $C$ must 
 occur exactly $\binom{r-1}{k-1}$ times. 
 
 Similarly, a $(k,0)$-coloring
 of $\kneser{r}{k}$ using less than $r$ colors also implies an independent set larger than  
 $\binom{r-1}{k-1}$, contradicting Theorem \ref{thm:ekr}.
\qed
\end{proof}

Now we prove Theorem  \ref{thm:natural}.

\begin{proof}[Proof of Theorem \ref{thm:natural}]
Let $r \geq 2k + 1$, $k \geq 1$ and $G = \kneser{r}{k}$.
WLOG let the colors used in $C$ be from the set $[r]$. We will show that 
there is a permutation function $\sigma: [r] \rightarrow [r]$ such that  every vertex $v$ 
with 
$C(v) =  \{c_1, c_2, \ldots, c_k\} \subseteq [r]$,
has natural coloring $N(v) = \{ \sigma(c_1), \sigma(c_2), \ldots, \sigma(c_k) \}$.

%

By Lemma \ref{lem:colorpat}, each color in $C$ must occur exactly $\binom{r-1}{k-1}$ times in $G$.
For any color $c \in [r]$, the vertices with color $c$ in $C$ form an independent set of size 
$\binom{r-1}{k-1}$. 
Now Theorem \ref{thm:hilton} states that every independent set of size $\binom{r-1}{k-1}$ is the set 
of vertices that contain some color $a$ in their natural coloring. 
Setting $a = \sigma(c)$, we get the desired mapping between the colors in $N$ and $C$.

To complete the proof, we need to show that the mapping $\sigma$ is a permutation. 
Since both domain and codomain of $\sigma$ is $[r]$, it is enough to show that 
$\sigma$ is injective.
Assume for the sake of contradiction that 
there are two colors $c_1 \neq c_2$ such that $\sigma(c_1) = \sigma(c_2) = a$.
Let $X$ be the set of vertices in $G$ that contain the color $a$ in the natural coloring $N$.
By our argument, note that $X$ is also the set of vertices in $G$ that contain the color $c_1$ in $C$ (and also the vertices
that contain the color $c_2$ in $C$).

Remove the set of vertices $X$ from $G$. 
What  remains is a copy 
of the Kneser graph $\kneser{r-1}{k}$, which is $(k,0)$-colored by $C$ using $r-2$ colors. 
Since $r-1 \geq 2k$, by Lemma \ref{lem:colorpat}, the graph $\kneser{r-1}{k}$ cannot be 
$(k,0)$-colored by $r-2$ colors. This contradicts the assumption that there are colors
$c_1 \neq c_2$ such that $\sigma(c_1) = \sigma(c_2) = a$. Hence $\sigma$ is injective and a permutation,
completing the proof.
\qed
\end{proof}

We conclude this section with one more definition, that of a totally independent 3-set in a Kneser graph.
\begin{definition}[Totally Independent 3-Set]
Consider the Kneser graph $\kneser{r}{k}$ with $r \geq 2k + 1$. Any set of three vertices of $\kneser{r}{k}$
form a \emph{totally independent 3-set} if they all share the set of same $k-1$ colors in their natural coloring, and differ
only in the last color.
\end{definition}
Because of Theorem \ref{thm:natural}, notice that 
the above definition remains unchanged if we use any $\qkicoloring{r}{k}{0}$-coloring
instead of the natural coloring.

\subsection{Proof of Theorem \ref{thm:npcing}}
In this section, we prove Theorem \ref{thm:npcing} which states that the $\qkicoloring{2k + i}{k}{0}$-coloring 
problem is NP-complete for all $k, i \geq 1$.

As before, it is easy to see that the problem is in NP and we shall focus on showing that 
the problem is NP-hard. 
We will show that 
there is a polynomial time reduction from 3-\emph{CNFSAT}  to $\qkicoloring{2k+i}{k}{0}$-coloring. 
Given a 3-CNF Boolean formula $\psi$, 
we will construct a graph $G$ in polynomial time such that $G$ is 
$\qkicoloring{2k+i}{k}{0}$-colorable if and only if $\psi$ is satisfiable.
The formula $\psi$ is a conjunction of clauses, with each clause consisting
of exactly 3 literals.
For a fixed value of $k$ and $i$, we now describe the construction of the graph $G$ from $\psi$. 

\medskip
\noindent\textbf{Construction of $G$ from $\psi$:} Let $\psi$ be a 3-CNF formula with $n$ variables 
and $m$ clauses. We first describe the vertices of the graph $G$.
\begin{enumerate}[V (i)]
\item Two vertices $u$ and $w$.

\item A set of $\binom{2k+i-1}{k-1}$ vertices denoted as follows: 
$$A = \left\{v_{\ell} \mid {\ell}=1,2,3,\ldots,{\binom{2k+i-1}{k-1}} \right\}.$$

\item For each variable $p$ in $\psi$, a set $B_p$ of $\binom{2k+i-1}{k}$ vertices. Each set $B_p$ is defined as follows:
$$B_p = \{x_p, \overline{x_p}\} \cup  \left\{y_{p, \ell} \mid {\ell}=1,2,3,\ldots, {\binom{2k+i-1}{k}} -2 \right \}.$$

\item For each clause $C_j$ in $\psi$, a set of three vertices denoted by $z_{j,1}, z_{j,2},$ and $z_{j,3}$. 

\item For each clause $C_j$ in $\psi$, a set of ${\binom{2k+i}{k}}$ vertices denoted by $\Gamma_{j}$ defined as follows:
$$\Gamma_j = \left\{\gamma_{j,\ell} \mid {\ell}=1,2,3,\ldots,{\binom{2k+i}{k}} \right\}.$$
\end{enumerate}
Thus the total number of vertices in $G$ is $2 +  {\binom{2k+i-1}{k-1}} + n {\binom{2k+i-1}{k}} + m \left(3 + 
{\binom{2k+i}{k}}\right)$. For fixed values of $k$ and $i$, the number of vertices in $G$ is polynomial
in the size of the formula $\psi$.

Now we describe the edges of $G$. While describing the edges, we use the shorthand notation $uw$ for the edge
$\{u, w\}$.
\begin{enumerate}[E (i)]
\item For each $1 \leq p \leq n$, the vertices $A \cup B_p$ form a copy of the
Kneser graph $\kneser{2k+i}{k}$.
When the vertices of $\kneser{2k+i}{k}$ are regarded as the $k$-sized subsets of $[2k+i]$, the
vertices in $A$ correspond to the $k$-sized subsets of 
$[2k+i]\setminus \{2k,2k+1\}$ and $(k-2)$-sized subsets of $[2k+i]\setminus \{2k, 2k+1\}$ 
union with $\{2k,2k+1\}$. That is, vertices in $A$ correspond to the $k$-sized subsets of 
$[2k+i]$ that either contain both $2k$ and $2k+1$, or contain none of $2k$ and $2k+1$.

The vertices $x_p$ and $\overline{x_p}$ correspond to 
the sets $[k-1]\cup \{2k\}$ and 
$[k-1]\cup \{2k + 1\}$ in some order. 
The vertices in $B_{p}$ correspond to the other $k$-sized subsets of 
$[2k+i]$, that contain either $2k$ or $2k+1$ but not both.

\item Let $u', w' \in A$ be the vertices such that $u'$ represents the set $\{1, 2, \ldots, k\}$ and 
$w'$ represents the set $\{k, k+1, \ldots, 2k-1\}$.
Then $uw$, $uu'$ and $ww'$ are edges.

\item Let $U', W' \subseteq A$ defined as follows\footnote{
We remark that when $i=1$, the sets $U'$ and $W'$ are empty sets. However, the correctness of proof 
is maintained even in this case.}. The set $U'$ represents those set of vertices 
which have the set $\{1, 2, \ldots , k-1\}$ with the $k$-th element from $\{2k+2, 2k+3, \ldots, 2k+i\}$ while 
$W'$ represents those set of vertices which have the set $\{k, k+1, \ldots, 2k-2\}$ with the $k$-th element from 
$\{2k+2, 2k+3, \ldots, 2k+i\}$. Note that $|U'|=|W'|= i-1$.

Now join $u$ to all the vertices in $U'$ using edges, and similarly join $w$ to all the vertices in $W'$.
\footnote{We can reduce the sizes of both $U'$ and $W'$ to $\lceil (i-1)/k\rceil$. This can be done by forming the $k$-sized sets from 
$\{2k+2, 2k+3, \ldots, 2k+i\}$  in increasing order. However, we use the simpler (but larger) sets $U'$ and $W'$ in the proof.}

\item For each $1 \leq j \leq m$, $wz_{j,1}, wz_{j,2}$  and $wz_{j,3}$ are edges.

\item Suppose the $j$-th clause is $C_j = x_{j,1} \vee x_{j,2} \vee x_{j,3}$, where each of $x_{j,1}$, $x_{j,2}$ and $x_{j,3}$ are literals 
$x_p$ or $\overline{x_p}$ for some $1 \leq p \leq n$. Then $z_{j,1}\overline{x_{j,1}}$, $z_{j,2}\overline{x_{j,2}}$ and $z_{j,3}\overline{x_{j,3}}$
are edges, for each $j$.

\item For each $j$, the vertex set $\Gamma_{j}$ forms a copy of $\kneser{2k+i}{k}$.

\item For each $j$, we identify three vertices $t_{j,1}$, $t_{j, 2}$ and $t_{j,3}$ from $\Gamma_j$, such that 
these three vertices form a totally independent 3-set in $\Gamma_j$. 

Then $z_{j,1}t_{j,1}$, 
$z_{j,2}t_{j,2}$ and  $z_{j,3}t_{j,3}$ are edges for each $j$.

\item For each $j$, all the vertices in $\{z_{j,1}, z_{j, 2}, z_{j, 3}\}$ are joined to all the vertices in $U'$, forming a complete bipartite graph.

\item Similarly, for each $j$, the sets $\{t_{j,1}, t_{j, 2}, t_{j, 3}\}$ and $W'$ form a complete bipartite graph.
\end{enumerate}

\begin{figure}[t!]
\centering
\begin{tikzpicture}
     [fill opacity=1, scale = 0.9]


 \node (Gamma) [draw,circle, minimum size=3.5 cm,label=above left:$\Gamma_j$] at (2, 11) {};

    \path coordinate [label = above left:$t_{j,1}$]  (i1) at (2.75,11.75)

          coordinate [label = above left:$t_{j,2}$] (i2) at (2,11)

         coordinate [label = above left:$t_{j,3}$] (i3) at (1.25,10.25)
          coordinate [label = above:$z_{j,1}$] (z1) at (9.25,11.75)
          coordinate [label = above:$z_{j,2}$] (z2) at (10,11)
          coordinate [label = above:$z_{j,3}$] (z3) at (10.75,10.25);

  \node (A) [draw,ellipse, minimum height=4 cm,minimum width= 7.5 cm, label=160:$A$] at (5, 6) {};

 \node (capwdash) [draw,  shade, ellipse, minimum height=1 cm,minimum width= 2 cm, label=above right:$W'$] at (3, 6) {};
 \node (capudash) [draw,  shade, ellipse, minimum height=1 cm,minimum width= 2 cm, label=above left:$U'$] at (7, 6) {};
    \path coordinate (a1) at (2.5,6)
          coordinate (a2) at (3.5,6)
          coordinate (b1) at (6.5,6)
          coordinate (b2) at (7.5,6)
          coordinate [label = {[label distance = 0.1 cm] above:$w'$}] (wdash) at (4.5,5)
          coordinate [label = {[label distance = 0.1 cm] above:$u'$}] (udash) at (5.5,5);
          \path (b1) -- node[auto=false]{\ldots} (b2);
          \path (a1) -- node[auto=false]{\ldots} (a2);

    \path coordinate [label = {[label distance = 0.05 cm] below:$u$}]  (u) at (7,2)
    	  coordinate [label = {[label distance = 0.05 cm] below:$w$}] (w) at (3,2)
	  coordinate (temp1) at (8.7, 2.7)
	  coordinate (temp2) at (9.0, 2.4)
	  coordinate (temp3) at (9.3, 2.1);

    \node (B1) [draw,circle, minimum size=2 cm,label=160:$B_1$] at (14, 14) {};
   
    \path coordinate [label = {[label distance = 0 cm] above:$x_1$}] (x1) at (13.5,14)
          coordinate [label = {[label distance = 0 cm] above:$\overline{x_1}$}] (x1b) at (14.5,14);
          
    \node (B2) [draw,circle, minimum size=2 cm,label=160:$B_2$] at (14, 11) {};
   
    \path coordinate [label = {[label distance = 0 cm] above:$x_2$}] (x2) at (13.5,11)
          coordinate [label = {[label distance = 0 cm] above:$\overline{x_2}$}] (x2b) at (14.5,11);

    \node (B6) [draw,circle, minimum size=2 cm,label=160:$B_6$] at (14, 7) {};
   
    \path coordinate [label = {[label distance = 0 cm] above:$x_6$}] (x6) at (13.5,7)
          coordinate [label = {[label distance = 0 cm] above:$\overline{x_6}$}] (x6b) at (14.5,7);

    \node (Bn) [draw,circle, minimum size=2 cm,label=160:$B_n$] at (14, 3) {};
   
    \path coordinate [label = {[label distance = 0 cm] above:$x_n$}] (xn) at (13.5,3)
          coordinate [label = {[label distance = 0 cm] above:$\overline{x_n}$}] (xnb) at (14.5,3);

          \path (B2) -- node[auto=false]{\vdots} (B6);
          \path (B6) -- node[auto=false]{\vdots} (Bn);

	\draw[color=black] (z1) to [out=-65,in=45] (temp1) ;
	\draw[color=black] (temp1) to  [out=-135,in=-45] (w) ;
	\draw[color=black] (z2) to [out=-65,in=45] (temp2) ;
	\draw[color=black] (temp2) to  [out=-135,in=-45] (w) ;
	\draw[color=black] (z3) to [out=-65,in=45] (temp3) ;
	\draw[color=black] (temp3) to  [out=-135,in=-45] (w) ;

    \draw (i1)--(z1) (i2)--(z2) (i3)--(z3); 
    \draw (i1)--(a1) (i1)--(a2) (i2)--(a1) (i2)--(a2) (i3)--(a1) (i3)--(a2);
    \draw (z1)--(b1) (z1)--(b2) (z2)--(b1) (z2)--(b2) (z3)--(b1) (z3)--(b2);
    \draw (u)--(w) (w)--(a1) (w)--(a2) (u)--(b1) (u)--(b2);
    \draw (w)--(wdash) (u)--(udash);
    \draw (z1)--(x1b);
    \draw (z2)--(x2);
    \draw (z3)--(x6b);

    \foreach \point in {a1,a2,b1,b2,i1,i2,i3,z1,z2,z3,u,w,wdash,udash,x1,x1b,x2,x2b,x6,x6b,xn,xnb}
      \fill [black,opacity=1] (\point) circle (1.5pt);
 \end{tikzpicture} 
 \caption{The graph $G$ that corresponds to the formula $\psi$.
 The figure depicts the edges of the graph $G$ for the clause $C_j= x_1 \vee \overline{x_2} \vee x_{6}$. For each clause in $\psi$, 
there will be a copy of vertices $\Gamma_{j}$, $z_{j,1}$, $z_{j,2}$ and $z_{j,3}$ and the
 corresponding edges. For reducing clutter, several vertices and edges have not 
 been shown in the figure.} \label{fig:red}
\end{figure}
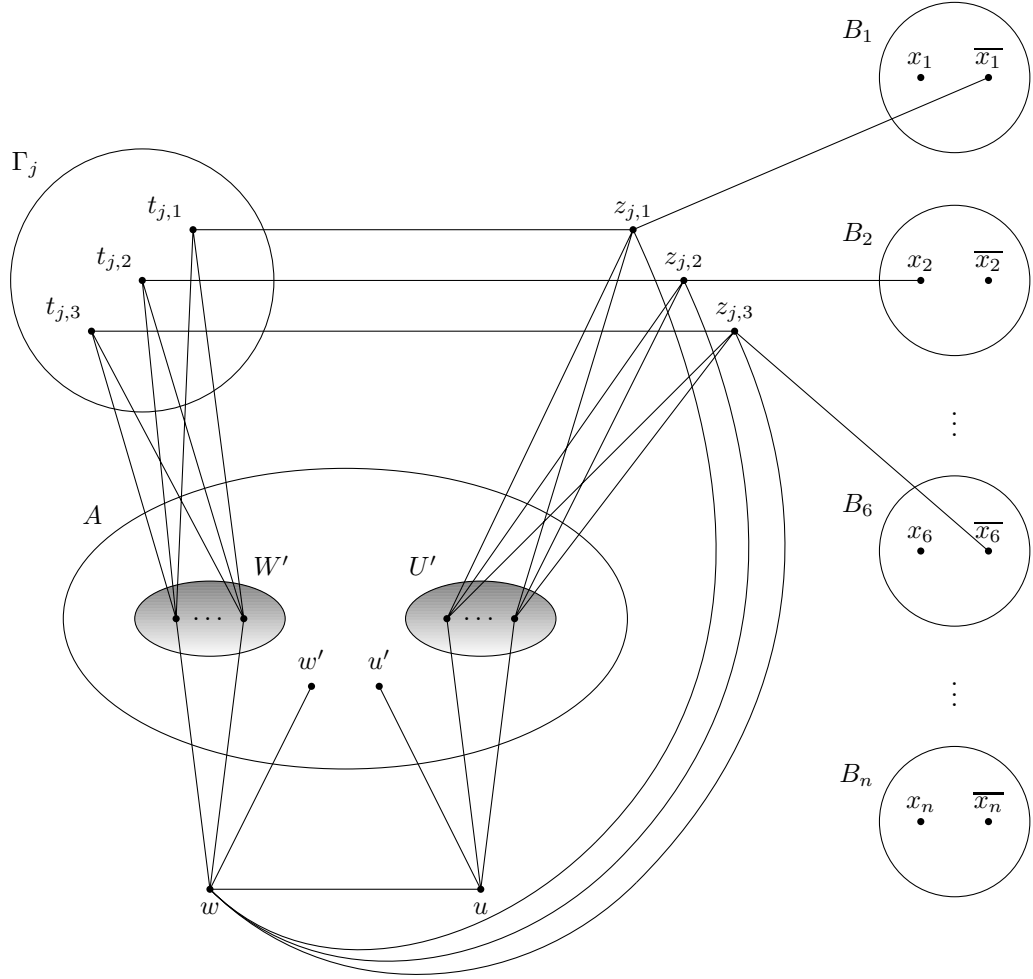
The description of the graph $G$ is complete. A pictorial representation is given in Figure \ref{fig:red}.
We now claim that the graph $G$ so constructed is $\qkicoloring{2k+i}{k}{0}$-colorable if and only if $\psi$ is satisfiable.

\medskip
\noindent\textbf{Colorability implies satisfiability:}
Suppose that $G$ is $\qkicoloring{2k+i}{k}{0}$-colorable. We will construct a satisfying assignment for $\psi$. 
Given a $\qkicoloring{2k+i}{k}{0}$-coloring of $G$, we can assume that the vertices of the set $A$ are colored according to the 
natural coloring of the Kneser graph $A \cup B_p$ for any $p$ using Theorem \ref{thm:natural}.
Using the vertices of $A$, the correspondence between the colors in the natural coloring and the given coloring is determined, 
except for the colors $2k$ and $2k+1$.

The vertex $u$ is adjacent to $u'$, which is colored $\{1, 2, \ldots, k\}$ and to the set $U'$ whose vertices are colored 
with the colors $\{1, 2, \ldots , k-1\} \cup \{2k+2, 2k+3, \ldots, 2k+i\}$. This leaves the colors $\{k+1, k+2, \ldots, 2k, 2k+1\}$
available
for coloring $u$. Similarly for $w$, because of adjacencies to $w'$ and $W'$, the available colors are 
$\{1, 2, \ldots, k-1\} \cup \{2k, 2k+1\}$. Further, because $uw$ is an edge, the colors $\{k+1, k+2, \ldots, 2k-1\}$ are 
fixed for $u$
and $\{1, 2, \ldots, k-1\}$ for $w$, 
with $u$ and $w$ being assigned one color each from $\{2k, 2k+1\}$.
Without loss of generality, we may assume that $u$ gets the color $2k$
and $w$ gets the color $2k + 1$. 
We will further refer to $2k$ as the `true color' and $2k+1$ as the `false color'.

{Given the way $A$ is colored, there are two ways in which each $B_p$ can be colored consistent with 
the coloring of $A$. In each $B_p$, the vertices $x_p$ can be colored $\{1, 2, \ldots, k-1, 2k\}$ and $\overline{x_p}$ 
be colored $\{1, 2, \ldots, k-1, 2k+1\}$ or vice versa. The coloring of the rest of the vertices in $B_p$ can be determined
from the coloring of $x_p$ and $\overline{x_p}$. 
The coloring of the vertices $x_p$ and $\overline{x_p}$ in the given coloring of $G$ will help us construct the 
satisfying assignment for the formula $\psi$.
If the vertex $x_p$ is colored $\{1, 2,\ldots,k-1, 2k\}$, then the literal $x_p$ is 
considered to be a true literal (and $\overline{x_p}$ false) since it has been assigned the `true color' $2k$. 
Similarly, if the vertex $x_p$ gets colored $\{1, 2,\ldots, k-1, 2k+1\}$, then $x_p$ is considered false. We will now 
see that the assignment thus obtained for each $x_p$ constitutes a satisfying assignment for $\psi$.}

Assume for the sake of contradiction that the assignment is not a satisfying assignment for $\psi$. 
This means that there is a clause $C_j =x_{j,1}\vee x_{j,2}\vee x_{j,3}$ that is not satisfied. 
Here each literal $x_{j,1}, x_{j,2}$ and $x_{j,3}$ correspond to some $x_p$ or $\overline{x_p}$. 
Since $C_j$ is not satisfied, all of $x_{j,1}, x_{j,2}$ and $x_{j,3}$ are false literals. This means 
that $\overline{x_{j,1}}, \overline{x_{j,2}}$ and $\overline{x_{j,3}}$ are all true literals and are all
colored with the set $\{1, 2,\ldots, k-1, 2k\}$. 

Now let us consider the vertices $z_{j, 1}, z_{j, 2}$ and $z_{j, 3}$. Because each of these vertices
is adjacent to $w$ and the set $U'$, the free colors available for $z_{j, 1}, z_{j, 2}$ and $z_{j, 3}$ are
$\{k, k+1, \ldots, 2k\}$. Since $z_{j, 1}, z_{j, 2}$ and $z_{j, 3}$ are adjacent to true literals, the color $2k$
is also ruled out, fixing the colors of $z_{j, 1}, z_{j, 2}$ and $z_{j, 3}$ to the set $\{k, k+1, \ldots, 2k-1\}$.

Now we focus on the vertices $t_{j, 1}, t_{j, 2}$ and $t_{j, 3}$. Because of its adjacencies to the set $W'$,
the free colors available to $t_{j, 1}, t_{j, 2}$ and $t_{j, 3}$ are 
$[k-1] \cup \{2k-1, 2k, 2k+1\}$.
Because each $t_{j, \ell}$ is adjacent to the corresponding $z_{j, \ell}$, the color $2k-1$ is ruled out, leaving
only the colors 
$[k-1] \cup \{2k, 2k+1\}$ 
for $t_{j, 1}, t_{j, 2}$ and $t_{j, 3}$. Recall that 
$t_{j, 1}, t_{j, 2}$ and $t_{j, 3}$ form a totally independent 3-set of the Kneser graph $\Gamma_j$.
Because of Theorem \ref{thm:natural}, the vertices $t_{j, 1}, t_{j, 2}$ and $t_{j, 3}$ together require $k+2$ colors,
since all of them share $k-1$ colors, with each getting a distinct $k$-th color. However, there are only $k+1$ free colors for 
$t_{j, 1}, t_{j, 2}$ and $t_{j, 3}$, namely the set 
$[k-1] \cup \{2k, 2k+1\}$. 
This is a contradiction 
to the assumption that there is a clause that is not satisfied. Hence the assignment is a satisfying assignment
for $\psi$.

\medskip
{\noindent\textbf{Satisfiability implies colorability:}
To prove the reverse direction, let us assume that $\psi$ is satisfiable. 
We will construct a $(2k+i, k, 0)$-coloring of $G$. 
We start with a satisfying assignment of $\psi$. 
For each $1 \leq p \leq n$, we color the vertices of the Kneser graph $A \cup B_p$ with their natural coloring.
If the variable $x_p$ is true in the satisfying assignment, we retain this coloring. If for some $p$, $x_p$ is assigned 
false, then for that $B_p$, we swap the colors $2k$ and $2k+1$ for all the vertices. Notice that this is still a valid coloring
of the Kneser graph $A\cup B_p$, because the vertices in $A$ will be unaffected by such a swap. For each variable $x_p$,
this coloring assigns 
the vertex $x_p$ with  $\{1, 2, \ldots, k-1, 2k\}$ and 
the vertex $\overline{x_p}$ with  $\{1, 2, \ldots, k-1, 2k+1\}$
 if $x_p$ is assigned true in the satisfying assignment, and vice versa if 
$x_p$ is assigned false. }

We assign vertex $u$  the colors $\{k+1, k+2,\ldots, 2k\}$  and $w$  the colors $\{1, 2, \ldots, k-1, 2k+1\}$.
Now let us consider the vertices $z_{j, 1}, z_{j, 2}$ and $z_{j, 3}$. If any of the $z_{j, \ell}$ is adjacent to a true literal
$x_p$ or $\overline{x_p}$, the vertex $z_{j, \ell}$ is assigned $\{k, k+1, \ldots, 2k-1\}$, since they are the only free colors
available. However, if any of the $z_{j, \ell}$ is adjacent to a false literal, then colors from $\{k, k+1, \ldots, 2k\}$ are 
available. 

Since we started with a satisfying assignment, every clause $C_j$ contains at least one true literal, WLOG say $x_{j, 1}$. 
The vertex $z_{j, 1}$ is adjacent to the vertex $\overline{x_{j, 1}}$ which is false. 
Now $z_{j, 1}$ can be assigned the set $\{k, k+1, \ldots, 2k-2, 2k\}$ while $z_{j,2}$ and $z_{j, 3}$ can be assigned 
the set $\{k, k+1, \ldots, 2k-1\}$.

The vertex $t_{j, 1}$ has the colors 
$[k-1] \cup \{2k-1, 2k+1\}$ 
available while the
vertices $t_{j, 2}$ and $t_{j, 3}$ have the colors 
$[k-1] \cup \{2k, 2k+1\}$ 
available for coloring.
We can assign the sets 
$\{1, 2, \ldots,k-1, 2k-1\}$, $\{1, 2, \ldots,k-1, 2k\}$ and $\{1, 2, \ldots,k-1, 2k+1\}$
to the vertices $t_{j, 1}, t_{j, 2}$ and $t_{j, 3}$ respectively. This is consistent with the coloring 
of a totally independent 3-set and hence can be extended to a $(k,0)$-coloring of $\Gamma_j$.

Now all the vertices have been assigned colors, giving a $(k, 0)$-coloring of $G$.

We have completed the proof of correctness of the reduction. We have already seen that 
the size of $G$ is polynomial
in the size of $\psi$. Thus we have shown that the $\qkicoloring{2k + i}{k}{0}$-coloring 
problem is NP-complete for all $k, i \geq 1$.
\qed

\medskip
\noindent \textbf{Acknowledgment:} The authors would like to thank the anonymous reviewer for helpful comments, and pointing
out a flaw in the proof of Theorem \ref{thm:npkminusone} in an earlier version of the paper.

%% file: main.bbl
\begin{thebibliography}{10}
\expandafter\ifx\csname url\endcsname\relax
  \def\url#1{\texttt{#1}}\fi
\expandafter\ifx\csname urlprefix\endcsname\relax\def\urlprefix{URL }\fi
\expandafter\ifx\csname href\endcsname\relax
  \def\href#1#2{#2} \def\path#1{#1}\fi

\bibitem{caldam2018}
S.~Joshi, S.~Kalyanasundaram, A.~S. Kare, S.~Bhyravarapu, On the tractability
  of $(k,i)$-coloring, in: B.~Panda, P.~P. Goswami (Eds.), Algorithms and
  Discrete Applied Mathematics: Fourth International Conference, CALDAM 2018,
  India, 2018, pp. 188--198.

\bibitem{IMDiaz1999}
I.~M\'{e}ndez-D\'{i}az, P.~Zabala, A generalization of the graph coloring
  problem, Investigation Operation 8 (1999) 167--184.

\bibitem{stahl1976}
S.~Stahl, $n$-tuple colorings and associated graphs, Journal of Combinatorial
  Theory, Series B 20~(2) (1976) 185 -- 203.

\bibitem{Bollobas1979}
B.~Bollob\'{a}s, A.~Thomason,
  \href{http://www.sciencedirect.com/science/article/pii/0012365X79901481}{Set
  colourings of graphs}, Discrete Mathematics 25~(1) (1979) 21 -- 26.
\newblock \href
  {http://dx.doi.org/https://doi.org/10.1016/0012-365X(79)90148-1}
  {\path{doi:https://doi.org/10.1016/0012-365X(79)90148-1}}.
\newline\urlprefix\url{http://www.sciencedirect.com/science/article/pii/0012365X79901481}

\bibitem{klostermeyer2002}
W.~Klostermeyer, C.~Q. Zhang, $n$-tuple coloring of planar graphs with large
  odd girth, Graphs and Combinatorics 18~(1) (2002) 119--132.

\bibitem{sparl2007}
P.~{\v{S}}parl, J.~{\v{Z}}erovnik, A note on $n$-tuple colourings and circular
  colourings of planar graphs with large odd girth, International Journal of
  Computer Mathematics 84~(12) (2007) 1743--1746.

\bibitem{Irving1983}
R.~W. Irving, {NP}-completeness of a family of graph-colouring problems,
  Discrete Applied Mathematics 5~(1) (1983) 111 -- 117.

\bibitem{marx2004}
D.~Marx, Graph colouring problems and their applications in scheduling,
  Periodica Polytechnica Electrical Engineering 48~(1-2) (2004) 11--16.

\bibitem{calvin2014}
C.~Beideman, J.~Blocki, Set families with low pairwise intersection, arXiv
  preprint arXiv:1404.4622.

\bibitem{brigham1982}
R.~C. Brigham, R.~D. Dutton, Generalized $k$-tuple colorings of cycles and
  other graphs, Journal of Combinatorial Theory, Series B 32~(1) (1982) 90--94.

\bibitem{Bonomo2014}
F.~Bonomo, G.~Dur\'an, I.~Koch, M.~Valencia-Pabon, On the $(k,i)$-coloring of
  cacti and complete graphs, Ars Combinatoria.

\bibitem{johnson1962}
S.~Johnson, A new upper bound for error-correcting codes, IRE Transactions on
  Information Theory 8~(3) (1962) 203--207.

\bibitem{Diptapriyo2017}
D.~Majumdar, R.~Neogi, V.~Raman, P.~Tale, Exact and parameterized algorithms
  for $(k, i)$-coloring, in: Algorithms and Discrete Applied Mathematics: Third
  International Conference, CALDAM 2017, India, 2017, pp. 281--293.

\bibitem{kratsch2010}
S.~Kratsch, P.~Schweitzer, Isomorphism for graphs of bounded feedback vertex
  set number, Algorithm Theory-SWAT 2010 (2010) 81--92.

\bibitem{jansen2014}
B.~M. Jansen, V.~Raman, M.~Vatshelle, Parameter ecology for feedback vertex
  set, Tsinghua Science and Technology 19~(4) (2014) 387--409.

\bibitem{fvsapprox}
V.~Bafna, P.~Berman, T.~Fujito, A 2-approximation algorithm for the undirected
  feedback vertex set problem, SIAM Journal on Discrete Mathematics 12~(3)
  (1999) 289--297.

\bibitem{inapproxtw}
Y.~L. Wu, P.~Austrin, T.~Pitassi, D.~Liu,
  \href{http://dl.acm.org/citation.cfm?id=2655713.2655729}{Inapproximability of
  treewidth, one-shot pebbling, and related layout problems}, J. Artif. Int.
  Res. 49~(1) (2014) 569--600.
\newline\urlprefix\url{http://dl.acm.org/citation.cfm?id=2655713.2655729}

\bibitem{fastfvs}
T.~Kociumaka, M.~Pilipczuk,
  \href{http://www.sciencedirect.com/science/article/pii/S0020019014000969}{Faster
  deterministic feedback vertex set}, Information Processing Letters 114~(10)
  (2014) 556 -- 560.
\newblock \href {http://dx.doi.org/https://doi.org/10.1016/j.ipl.2014.05.001}
  {\path{doi:https://doi.org/10.1016/j.ipl.2014.05.001}}.
\newline\urlprefix\url{http://www.sciencedirect.com/science/article/pii/S0020019014000969}

\bibitem{Fomin2010}
F.~V. Fomin, D.~Kratsch, Exact Exponential Algorithms, Texts in Theoretical
  Computer Science, An EATCS Series. Springer, 2010.

\bibitem{moonmoser}
J.~W. Moon, L.~Moser, \href{https://doi.org/10.1007/BF02760024}{On cliques in
  graphs}, Israel Journal of Mathematics 3~(1) (1965) 23--28.
\newblock \href {http://dx.doi.org/10.1007/BF02760024}
  {\path{doi:10.1007/BF02760024}}.
\newline\urlprefix\url{https://doi.org/10.1007/BF02760024}

\bibitem{Koivisto}
M.~Koivisto, \href{http://dx.doi.org/10.1109/FOCS.2006.11}{An ${O}^*(2^n)$
  algorithm for graph coloring and other partitioning problems via
  inclusion--exclusion}, in: Proceedings of the 47th Annual IEEE Symposium on
  Foundations of Computer Science, FOCS '06, IEEE Computer Society, Washington,
  DC, USA, 2006, pp. 583--590.
\newblock \href {http://dx.doi.org/10.1109/FOCS.2006.11}
  {\path{doi:10.1109/FOCS.2006.11}}.
\newline\urlprefix\url{http://dx.doi.org/10.1109/FOCS.2006.11}

\bibitem{ekr}
P.~Erd\H{o}s, C.~Ko, R.~Rado,
  \href{http://dx.doi.org/10.1093/qmath/12.1.313}{Intersection theorems for
  systems of finite sets}, The Quarterly Journal of Mathematics 12~(1) (1961)
  313--320.
\newblock \href {http://dx.doi.org/10.1093/qmath/12.1.313}
  {\path{doi:10.1093/qmath/12.1.313}}.
\newline\urlprefix\url{http://dx.doi.org/10.1093/qmath/12.1.313}

\bibitem{hm}
A.~J.~W. Hilton, E.~C. Milner,
  \href{http://dx.doi.org/10.1093/qmath/18.1.369}{Some intersection theorems
  for systems of finite sets}, The Quarterly Journal of Mathematics 18~(1)
  (1967) 369--384.
\newblock \href {http://dx.doi.org/10.1093/qmath/18.1.369}
  {\path{doi:10.1093/qmath/18.1.369}}.
\newline\urlprefix\url{http://dx.doi.org/10.1093/qmath/18.1.369}

\end{thebibliography}
